\documentclass[amsmath,amssymb,amsthm]{revtex4}

\usepackage{graphicx}
\usepackage{dcolumn}
\usepackage{bm}
\usepackage{amsmath}
\usepackage{amssymb}
\usepackage{amsthm}

\newtheorem{thm}{Theorem}[section]
\newtheorem{lem}{Lemma}[section]
\newtheorem{defi}{Definition}

\newcommand{\be}{\begin{equation}}
\newcommand{\ee}{\end{equation}}

\newcommand{\fa}{\Psi_A}

\newcommand{\fad}{\Psi_A^\dagger}

\newcommand{\lam}{\lambda}
\newcommand{\ra}{\rangle}
\newcommand{\la}{\langle}

\begin{document}

\title{One-Dimensional Impenetrable Anyons in Thermal Equilibrium. I.\\
                   Anyonic Generalization of  Lenard's  Formula}

\author{Ovidiu I. P\^{a}\c{t}u }
   \affiliation{C.N. Yang Institute for Theoretical Physics, State University of
   New York at Stony Brook, Stony Brook, NY 11794-3840, USA  }
   \altaffiliation{Permanent address: Institute for Space Sciences MG 23, R 077125
   Bucharest-M\u{a}gurele, Romania} \email{ipatu@grad.physics.sunysb.edu}

\author{Vladimir E. Korepin}
    \affiliation {C.N. Yang Institute  for Theoretical Physics, State University
    of New York at Stony Brook, Stony Brook, NY 11794-3840, USA}
    \email{korepin@max2.physics.sunysb.edu}

\author{Dmitri V. Averin}
    \affiliation{Department of Physics and Astronomy, State University of New York
    at Stony Brook, Stony Brook, NY 11794-3800, USA}
    \email{dmitri.averin@stonybrook.edu}

\begin{abstract}

We have obtained an expansion of the reduced density matrices (or,
equivalently, correlation functions of the fields) of impenetrable
one-dimensional anyons in terms of the reduced density matrices of
fermions using the mapping between anyon and fermion wavefunctions.
This is the generalization to anyonic statistics of the result
obtained by A. Lenard \cite{L1} for bosons. In the case of
impenetrable but otherwise free anyons with statistics parameter
$\kappa$, the anyonic reduced density matrices in the grand
canonical ensemble is expressed as Fredholm minors of the integral
operator ($1-\gamma \hat \theta_T$) with complex
statistics-dependent coefficient $\gamma=(1+e^{\pm i\pi\kappa })/
\pi$. For $\kappa=0$ we recover the bosonic case of Lenard
$\gamma=2/\pi$. Due to nonconservation of parity, the anyonic field
correlators $\la \fad(x')\fa(x)\ra$ are different depending on the
sign of $x'-x$.

\end{abstract}

\maketitle

\section{Introduction}

During the last few years, several models of one-dimensional anyons
\cite{LMP,IT,Kundu,G2,AN} have attracted considerable attention.
This attention is motivated, besides general fundamental physics
interest of anyons, by several new statistical-mechanical features
of the anyonic systems. One of them is the behavior of the
zero-temperature field correlation functions, which exhibit
oscillations with the period dependent on the statistics parameter
in the leading term of their algebraically decaying large-distance
asymptotics. These oscillations were obtained within the harmonic
fluid approach \cite{CM} and conformal field theory \cite{PKA}.

Much of this effort was focused on the anyonic model proposed in
\cite{Kundu} and clarified in \cite{AN,PKA}. This model can be
understood as an anyonic extension of the Bose gas with
$\delta$-function interaction \cite{LL} which was solved by Lieb and
Liniger. The study of this model was initiated by Batchelor {\it et
al.} \cite{BGO,BG,BGH} and continued in \cite{SSC,CM,PKA}. In the
limit of infinite interaction strength, the anyons become
impenetrable. The present work is the first paper in a series that
intends to provide a complete treatment of the correlation functions
of such impenetrable free anyons in a manner similar to the
impenetrable Bose gas \cite{KBI}. The first step taken in this work
is the anyonic generalization of Lenard's formula \cite{L1} which
gives the reduced density matrices of hard-core bosons in terms of
the fermionic reduced density matrices. With this result, the
reduced density matrices of free impenetrable bosons can be
expressed in terms of Fredholm minors of the integral operator
$\hat\theta_T$ with the kernel given by the Fourier transform of the
Fermi distribution function.

In this work, we use the Anyon-Fermi \cite{G2,AN} and Anyon-Bose
\cite{Kundu} mapping  to generalize the treatment of Lenard to the
case of anyons. Our main result is  Theorem \ref{main} which in the
case of impenetrable free anyons leads to a representation of the
correlation functions of anyonic fields as Fredholm minors of the
operator $\hat\theta_T$. This representation of the correlation
functions can be summarized as follows. The Fermi distribution function
$\vartheta(k,T,h)$ at temperature $T$ and chemical potential $h$
defines the integral operator $\hat\theta_T$ with the kernel
\be \theta_T(x-y)=\frac{1}{2}\int_{-\infty}^\infty dk\
e^{ik(x-y)}\vartheta(k,T,h),\ \ \;\;\;
\vartheta(k,T,h)=\frac{1}{1+e^{(k^2-h)/T}}\, , \ee
which acts on an arbitrary function $f(x)$ as
\be \left(\hat\theta_T f\right)(x)=\int_{I}\theta_T(x-y)f(y)\ dy\, ,
\ee
where $I$ is an interval or a finite union of intervals on the real
axis. The resolvent kernel $\varrho_T(x,y)$ associated with the kernel
$\theta_T(x,y)$ is defined to satisfy the equation
\be \varrho_T(x,y)-\gamma\int_I \theta_T(x-z)\varrho_T(z,y)\
dz=\theta_T(x-y)\, . \ee
Then, in the thermodynamic limit, the temperature-dependent
correlator of the anyonic field operators is given by the following
relations:
\begin{itemize}
\item If $x'>x$,
\be \la\fad(x')\fa(x)\ra_{T,h}=\frac{1}{\pi}\varrho_T(x',x)
\left.\det\left(1-\gamma\hat \theta_T\right)\right|_{\gamma=(1+e^{+
i\pi\kappa})/\pi}\, , \ee
where $\det\left(1-\gamma\hat \theta_T\right)$ is the Fredholm
determinant of the integral operator $\hat\theta_T$ (see Appendix
\ref{Fredholm}) which acts on the interval $I_+=[x,x']$.
\item
If $x'<x$, \be
\la\fad(x')\fa(x)\ra_{T,h}=\frac{1}{\pi}\varrho_T(x,x')
\left.\det\left(1-\gamma\hat \theta_T\right)\right|_{\gamma=(1+e^{-
i\pi\kappa})/\pi}\, , \ee
and now the integral operator acts on the interval $I_-=[x',x]$.
\end{itemize}
Using Eq.~(\ref{RK}) we can also express these formulae in terms of
the first Fredholm minor of the integral operator $\hat\theta_T$.

The methods used in our work allows for direct extension of these
results to the $2n$-point correlators with certain ordering of the
arguments:
\begin{itemize}
\item If $x_1<x_1'<\cdots<x_n<x_n'$,
\be \la\fad(x_n')\cdots\fad(x_1')\fa(x_1)\cdots\fa(x_n)\ra_{T,h}=
\frac{C(x_1',\cdots,x_n)}{\pi^n} \varrho_T\left(\begin{array}{c}
x_1,\cdots,x_n\\ x_1',\cdots,x_n' \end{array}\right)
\left.\det\left(1-\gamma\hat
\theta_T\right)\right|_{\gamma=(1+e^{+i\pi\kappa})/\pi} , \ee
where now the integral operator acts in the region $I_+=[x_1,x_1']
\cup \cdots\cup[x_n,x_n']$, the statistics factor $C$ is defined by
Eqs.~(\ref{cc}), (\ref{aa}), and (\ref{bb}), and other notations are
given in Appendix \ref{Fredholm}.
\item If $x_1'<x_1<\cdots<x_n'<x_n$,
\be \la\fad(x_n')\cdots\fad(x_1')\fa(x_1)\cdots\fa(x_n)\ra_{T,h}=
\frac{C(x_1',\cdots,x_n)}{\pi^n} \varrho_T\left(\begin{array}{ccc}
x_1,\cdots,x_n\\ x_1',\cdots,x_n' \end{array}\right)
\left.\det\left(1-\gamma\hat \theta_T\right)\right|_{\gamma=(1+e^{-
i\pi\kappa})/\pi} , \ee
where the integral operator acts in the region
$I_-=[x_1',x_1]\cup\cdots\cup[x_n',x_n]$.
\end{itemize}
Again, Eq.~(\ref{RK}) can be used to express the $2n$-point
correlators as the $n$th Fredholm minors of the integral operator
$\hat\theta_T$.

The paper is organized as follows. In Sec.~\ref{sect1}, we introduce
the reduced density matrices of anyons and their relation to the
correlation functions of the anyonic fields. In Sec.~\ref{sect2}, we
describe the mappings between the wavefunctions of the
one-dimensional particles with Fermi, Bose, and Anyonic statistics.
In Sec.~\ref{sect3}, we prove the main theorem which expresses the
reduced density of anyons in terms of the reduced density matrices
of bosons and fermions. The consequences of this theorem in the case
of free impenetrable particles is analyzed in Sec.~\ref{sect4},
where we obtain the reduced density matrices of anyons as Fredholm
minors of an integral operator. The basic information on Fredholm
determinants, details of calculations with anyonic fields, and a
sketch of the derivation of the reduced density matrices for free
fermions in the grand canonical ensemble are presented in three
appendices.

\section{From Correlation Functions to Reduced Density Matrices} \label{sect1}

The one-dimensional anyons considered in this work are characterized
by anyonic fields $\fad(x), \fa(x)$ which obey the following
commutation relations
\be\label{com1}
\fa(x_1)\fad(x_2)=e^{-i\pi\kappa\epsilon(x_1-x_2)}\fad(x_2)\fa(x_1)+
\delta(x_1-x_2)\, ,
 \ee
\be\label{com2}
\fad(x_1)\fad(x_2)=e^{i\pi\kappa\epsilon(x_1-x_2)}\fad(x_2)\fad(x_1)\,
, \ee \be\label{com3}
\fa(x_1)\fa(x_2)=e^{i\pi\kappa\epsilon(x_1-x_2)}\fa(x_2)\fa(x_1)\, .
\ee
Here $\kappa$ is the statistics parameter, and $\epsilon(x)=x/|x|,\
\epsilon(0)=0$. The commutation relations become bosonic for
$\kappa=0$ and fermionic for $\kappa=1$. For an arbitrary
Hamiltonian of the anyons confined to the interval $V=[-L/2,L/2]$,
the $N$-particle eigenstates are defined as
\be\label{eq4e} |\Psi_N(\{\lam \})\ra=\frac{1}{\sqrt{N!}}\int_V
dz_1\cdots\int_V dz_N\ \chi^a_N(z_1,\cdots,z_N|\{\lam\})
\fad(z_N)\cdots\fad(z_1)|0\ra\, , \ee
\be\label{eq5e} \la\Psi_N(\{\lam \})|=\frac{1}{\sqrt{N!}}\int_V
dz_1\cdots\int_V dz_N\ \la 0|\fa(z_1)\cdots \fa(z_N) \chi^{*a}_N
(z_1,\cdots,z_N|\{\lam\})\, , \ee
where $\chi_N^a$ are the (norm one) quantum-mechanical wavefunctions
of $N$ anyons, and $\{\lam\}$ is a set of quantum numbers specifying
the state. Note the order of the field operators in these relations,
which is dictated by the boundary conditions on the wavefunctions in
the periodic or quasi-periodic case \cite{AN}. The wavefunctions
$\chi$ have the anyonic symmetry
\be \label{asdf}
\chi_N^a(z_1,\cdots,z_i,z_{i+1},\cdots,z_N)=e^{i\pi\kappa
\epsilon(z_i-z_{i+1})}\chi_N^a(z_1,\cdots,z_{i+1},z_{i},
\cdots,z_N)\, , \ee
that reflects the field commutation relations. We are interested in
computing the  finite-temperature correlation functions of anyonic
fields. The simplest example of these correlators is
\be \la\fad(x')\fa(x)\ra\, . \ee
In the grand canonical ensemble characterized by temperature $T$ and
chemical potential $h$, the field correlation function is given by
following relation
\be\label{2correlation}
\la\fad(x')\fa(x)\ra_{T,h}=\sum_{N=1}^\infty\sum_{\{\lam\}}e^{hN/T}
\frac{e^{-E(\{\lam\})/T}}{Z(h,V,T)}
\la\Psi_N(\{\lam\})|\fad(x')\fa(x)|\Psi_N(\{\lam\})\ra\, , \ee
where $E(\{\lam\})$ is the energy of the eigenstate with quantum
numbers $\{\lam\},$ and $Z(h,V,T)$ is the grand-canonical partition
function
\be Z(h,V,T)=\sum_{N=0}^\infty\sum_{\{\lambda\}}e^{hN/T}
e^{-E(\{\lambda\})/T}\, . \ee
As shown in the Appendix \ref{ac}, the correlator at  a fixed number of
particles $N$ is given by an overlap integral of the corresponding
wavefunction
\be \label{eq100}
\la\Psi_N(\{\lam\})|\fad(x')\fa(x)|\Psi_N(\{\lam\})\ra= N\int_V
dz_1\cdots\int_V dz_{N-1}\
\chi_N^{*a}(z_1,\cdots,z_{N-1},x'|\{\lam\})\chi_N^a(z_1,\cdots,
z_{N-1},x|\{\lam\})\, , \ee
so that the correlation function (\ref{2correlation}) can be written
as
\begin{eqnarray}\label{2point}
\la\fad(x')\fa(x)\ra_{T,h} &=& \sum_{N=1}^\infty
\sum_{\{\lam\}}e^{hN/T} \frac{e^{-E(\{\lam\})/T}}{Z(h,V,T)} N\int_V
dz_1\cdots\int_V dz_{N-1} \nonumber \\ & &
\times\chi_N^{*a}(z_1,\cdots,z_{N-1},x' |\{\lam\})
\chi_N^a(z_1,\cdots, z_{N-1},x|\{\lam\})\, .
\end{eqnarray}

We will also be interested in a class of $2n$-point field
correlation functions at finite temperature:
\be \la\fad(x_n')\cdots\fad(x_1')\fa(x_1)\cdots\fa(x_n)\ra_{T,h}\, .
\ee
These correlators can be expressed similarly to Eq.~(\ref{2point})
\begin{equation*}
\la\fad(x_n')\cdots\fad(x_1')\fa(x_1) \cdots\fa(x_n) \ra_{T,h}
=\sum_{N=n}^\infty\sum_{\{\lam\}}e^{hN/T} \frac{e^{-E(\{\lam\})/T}
}{Z(h,V,T)} \frac{N!}{(N-n)!} \int_Vdz_1\cdots\int_Vdz_{N-n}
\end{equation*}
\be\label{multi} \times \chi_N^{*a}(z_1,\cdots,z_{N-n},x_1',\cdots,
x_n'|\{\lam\})\chi_N^a(z_1,\cdots,z_{N-n},x_1,\cdots,x_n|\{\lam\})\,
. \ee

As one can see from (\ref{2point}) and  (\ref{multi}), the
correlation functions are obtained as a combination of the
wavefunctions and ensemble probabilities. In this context, it is
useful, similarly to the case of fermionic or bosonic particles, to
introduce the reduced density matrices of anyons:
\begin{defi}\label{definition}
For a statistical ensemble characterized by the probabilities
$p_{\{\lam\}}^N$,  the anyonic n-particle reduced density matrix is
defined as
\begin{equation*}
(x_1,\cdots,x_n|\rho^a_n|x_1',\cdots,x_n')= \sum_{N=n}^\infty
\sum_{\{\lam\}} p^N_{\{\lam\}}\frac{N!}{(N-n)!}\int_V dz_{1}\cdots\int_V
dz_{N-n}
\end{equation*}
\be \times\chi_{N}^{*a}(z_1,\cdots,z_{N-n},x_1',\cdots,x_n'|\{\lam\})
\chi_{N}^{a}(z_1,\cdots,z_{N-n},x_1,\cdots,x_n|\{\lam\})\, , \label{def}
\ee
where the wavefunctions $\chi_{N,\{\lam\}}^a$ are normalized to one.
\end{defi}

If the probabilities $p_{\{\lam\}}^N$ coincide with those in the grand
canonical ensemble, $p_{\{\lam\}}^N=e^{Nh/T} e^{-E(\{\lam\})/T}
/Z(h,V,T)$, the one-particle reduced density matrix is just the
2-point correlator (\ref{2point}), and the $n$-particle reduced
density matrix is the particular $2n$-point correlator
(\ref{multi}). These relations are exactly the same as in the case
of bosonic and fermionic statistics. A particular ``anyonic''
feature of the Definition \ref{definition}, is the fact that we
integrate over the first $N-n$ arguments of the wavefunctions. In
the case of bosonic and fermionic reduced density matrices,
integration over any subset of the $N-n$ out of $N$ arguments
produces the same result due to the parity of the wavefunctions.
This is not the case for the reduced density matrices of anyons due
to the anyonic symmetry (\ref{asdf}) which in general, e.g., in the
periodic or quasi-periodic situation (``anyons on a ring'') makes
different arguments of the wavefunctions inequivalent -- see
discussion in the next Section.

Under a certain set of conditions, which is also made precise in the
next Section, there is a correspondence between the anyonic and the
fermionic or bosonic wavefunctions. This correspondence will be used
later to express the anyonic reduced density matrices as expansions
in terms of fermionic or bosonic ones.

\section{Anyon-Fermi and Anyon-Bose mapping} \label{sect2}

To establish the correspondence between the wavefunctions of anyons
and fermions or bosons we define the two functions which essentially
incorporate statistical properties of the wavefunctions of different
statistics in one dimension:
\be A_\kappa(z_1,\cdots,z_N)=e^{i\pi\kappa\sum_{j<k}
\epsilon(z_j-z_k)/2} \label{aa} \ee
and \be \label{bb} B(z_1,\cdots,z_N)=\prod_{j>k}\epsilon(z_j-z_k)\,
, \ee
where the notations are the same as in Eq.~(\ref{com1}).

The mapping between anyons and fermions or bosons is analogous to
the Bose-Fermi mapping discovered in \cite{G}, where it was noticed
that any wavefunction of $N$ fermions has a bosonic counterpart
given by
\be\label{BF}
\chi^b(z_1,\cdots,z_N)=B(z_1,\cdots,z_N)\chi^f(z_1,\cdots,z_N)\, .
\ee
This correspondence is valid under very general conditions, with no
restrictions on the external or particle-particle interaction
potential, except for the requirement of the hard-core condition
which should make the bosons impenetrable. For particles confined to
a box with ``hard wall" boundary conditions (BC), bosonic and
fermionic wavefunctions satisfy the same BC. In this case, if
$\chi^f$ is an eigenfunction of the Hamiltonian, then $\chi^b$ is
also an eigenfunction with the same eigenvalue. However, in the case
of a ring of length $L$ with periodic BC for bosons, the BC for the
fermions are in general different and given by
\be \chi^f(0,\cdots,z_N)=(-1)^{N-1}\chi^f(L,\cdots,z_N)\, .
\label{bf} \ee
For even $N$, when Eq.~(\ref{bf}) means that the BC for fermions and
bosons are different by a phase shift $\pi$, the  relation between the
eigenenergies of the fermionic and bosonic systems is less direct.
Since for non-coincident coordinates $B^2\equiv 1$, the Bose-Fermi
mapping (\ref{BF}) is symmetric and remains true if the superscripts
$b$ and $f$ are interchanged.

\subsection{Anyon-Fermi Mapping}

It is straightforward to see that similarly to the Bose-Fermi
mapping (\ref{BF}), the wavefunction with anyonic symmetry
(\ref{asdf}) can be obtained by multiplication of a fermionic
wavefunction with the statistics factors (\ref{aa}) and (\ref{bb})
\cite{G2,AN}:
\be \label{AF}
\chi^a(z_1,\cdots,z_N)=A_\kappa(z_1,\cdots,z_N)B(z_1,\cdots,z_N)
\chi^f(z_1,\cdots,z_N)\, , \ee
Besides the anyonic symmetry (\ref{asdf}), the wavefunction $\chi^a$
(\ref{AF}) satisfies the condition
\be \chi^a(z_1,\cdots,z_N)|_{z_i=z_j}=0\,\ \ \mbox{ for all }
\{i,j\} \in \{1,\cdots,N\} \,. \label{hc} \ee
This means that the correspondence (\ref{AF}) is valid as long as
potential energy contains a hard-core part which ensures that the
anyons are impenetrable and condition (\ref{hc}) is indeed
satisfied. Other properties of the Anyon-Fermi mapping (\ref{AF})
are similar to those of the Bose-Fermi mapping. It is valid for an
arbitrary form of the potential energy in the particle Hamiltonian.
When the particles are confined to an interval with {\em ``hard
wall" boundary conditions}, the fermionic and anyonic systems both
satisfy the same BC of the wavefunctions vanishing at the ends of
the interval. In this case, if $\chi^f$ is an eigenfunction of the
particle Hamiltonian, then $\chi^a$ is also an eigenfunction of this
Hamiltonian with the same eigenvalue. This follows from the fact
that the statistics factors (\ref{aa}) and (\ref{bb}) are constant
everywhere except for the points of coincident coordinates, where
the wavefunctions vanish.

When particles are confined to a ring with {\em periodic or
quasi-periodic BC}, the properties of the Anyon-Fermi mapping are
more complicated. In this case, the anyonic wavefunction will have
different boundary conditions for each of its coordinate (see
\cite{AN} and Appendix A of \cite{PKA}), the difference being given
by an extra phase shift that depends on the statistics parameter
$\kappa$. Specifically, if the fermion wave function obeys some
generic quasi-periodic BC (the same in all coordinates) which can be
written as
\be \chi^f(0,\cdots,z_N)=(-1)^{N-1}e^{-i\phi}\chi^f(L,\cdots,z_N)\,
, \ee
then the anyonic wavefunction obeys the following BC in its
different arguments
\begin{eqnarray}
\chi^a(0,z_2,\cdots,z_N)& =& e^{-i\overline\phi}\ \ \chi^a(L,z_2,\cdots,z_N)
\, , \nonumber
\\ \chi^a(z_1,0,\cdots,z_N) &=&e^{i (2\pi\kappa -\overline\phi)}
\chi^a(z_1,L,\cdots,z_N) \, , \nonumber \\
                 & \vdots &   \\
\chi^a(z_1,z_2,\cdots,0)&=&e^{i(2(N-1) \pi\kappa -\overline\phi)}
\chi^a(z_1,z_2\cdots,L) \, , \nonumber
\end{eqnarray}
where $\overline\phi=\phi+\pi\kappa(N-1)$. As for the Bose-Fermi
mapping (\ref{BF}) with even $N$, the anyonic and fermionic
eigenenergies are not related directly in the situation of a ring
with quasiperiodic BC. In physics terms, this difference between
anyons and fermions corresponds to the statistical magnetic flux
$\pi \kappa (N-1)$ through the ring produced by $N$ one-dimensional
anyons of statistics $\kappa$.

The Anyon-Fermi mapping (\ref{AF}) is not symmetric. The inverse
relation can be written as
\be \chi^f(z_1,\cdots,z_N)=A_{-\kappa}(z_1,\cdots,z_N) B(z_1,\cdots,
z_N) \chi^a(z_1,\cdots,z_N)\, . \ee

\subsection{Anyon-Bose Mappping}

The Anyon-Bose mapping was historically the first mapping of this
kind introduced for one-dimensional anyons in \cite{Kundu}:
\be\label{AB}
\chi^a(z_1,\cdots,z_N)=A_\kappa(z_1,\cdots,z_N)\chi^b(z_1,\cdots,z_N)\,
. \ee
The wavefunction $\chi^a$ in (\ref{AB}) has the correct anyonic symmetry
(\ref{asdf}), and in contrast to Anyon-Fermi mapping (\ref{AF}),
need not vanish when a pair of coordinates
coincide. It should be noted, however, that without this hard-core
condition, the discontinuity of the statistics factor $A_\kappa$
(\ref{aa}) at coinciding coordinates translates into discontinuity
of the wavefunctions (\ref{AB}). In this case one needs an
additional condition regularizing the wavefunctions. Also, without
the hard-core condition, the statistics factor $A$ changes
substantially the behavior of the wavefunctions at the points of
coincident coordinates (for instance, if the particle-particle
interaction is $\delta$-functional, statistics renormalizes the
interaction strength \cite{PKA}) and the energy eigenvalues of the
bosonic and anyonic problems are different regardless of the
boundary conditions.

Other properties of the Anyon-Bose  mapping (\ref{AB}) are very
similar to those of the Anyon-Fermi mapping. It is valid for
arbitrary potential energy. For particles in a box with ``hard wall"
BC, both wavefunctions (\ref{AB}) satisfy the same condition
$\chi^a=0$ and $\chi^b=0$ at the boundary. For particles on a ring
with generic quasi-periodic BC for bosons that can be written as
\be \chi^b(0,\cdots,z_N)=e^{-i\phi }\chi^b(L,\cdots,z_N) \ee
(and have the same form for all other arguments of $\chi^b$), the
anyonic wavefunction obeys the following BC
\begin{eqnarray}
\chi^a(0,z_2,\cdots,z_N)& =& e^{-i\overline\phi}\ \ \ \
\chi^a(L,z_2,\cdots,z_N) \, , \nonumber
\\ \chi^a(z_1,0,\cdots,z_N) &=&e^{i(2\pi\kappa-\overline\phi)}
\chi^a(z_1,L,\cdots,z_N) \, , \nonumber \\
                 & \vdots &   \\
\chi^a(z_1,z_2,\cdots,0)&=&e^{i(2(N-1) \pi\kappa-\overline\phi)}
\chi^a(z_1,z_2\cdots,L) \, , \nonumber
\end{eqnarray}
where $\overline\phi=\phi+\pi\kappa(N-1)$.
The Anyon-Bose mapping is also not symmetric. The inverse of
(\ref{AB}) is
\be
\chi^b(z_1,\cdots,z_N)=A_{-\kappa}(z_1,\cdots,z_N)\chi^a(z_1,\cdots,z_N)\,
. \ee

\section{Main Theorem}\label{sect3}

In this section, we consider an arbitrary statistical ensemble in
which the states $\chi_{N,\{\lam\}}^a$ occur with probabilities
$p^N_{\{\lam\}}$. The anyons are assumed to be confined to an interval
$V=[-L/2,L/2]$, and wavefunctions are normalized to 1:
$\|\chi^a_{N,\{\lam\}}\|=1$. Our goal is to establish a relation between
the reduced density matrices $\rho_n^a$ of anyons (\ref{def}) and
similarly defined reduced density matrices of bosons $\rho_m^b$ and
fermions $\rho_p^f$. The fermionic and bosonic states that
correspond to the anyonic states $\chi_{N,\{\lam\}}^a$ in the
Anyon-Fermi (\ref{AF}) or Anyon-Bose (\ref{AB}) mappings have
similarly normalized wavefunctions, and we assume that they have the
same probabilities $p_{\{\lam\}}^N$. This assumption is natural under the
conditions (discussed in the previous Section) for which the
energies of the states of different statistics are the same, as they
are, for instance, when the wavefunctions satisfy the ``hard wall"
boundary conditions and hard-core condition on the particle-particle
interaction. The relation between the reduced density matrices is
established by the following theorem:

\begin{thm}\label{main}
Let $x_1,\cdots,x_n,x_1',\cdots,x_n'$ be $2n$ coordinates in the
interval $V$, and $O_{\pm}$ are the parts of the space of these
coordinates in which they are ordered, respectively, as
$x_1<x_1'<\cdots<x_n<x_n'$ and $x_1'<x_1<\cdots<x_n'<x_n$.  For the
$O_+$ ordering, one can define the subset of $V$:
$I_+=[x_1,x_1']\cup\cdots\cup[x_n,x_n']\subset V$, and the subset
$I_-=[x_1',x_1]\cup \cdots \cup[x_n',x_n]\subset V$ for ordering as
in $O_-$. If the conditions of validity of the Anyon-Fermi
(\ref{AF}) or the Anyon-Bose mapping (\ref{AB}) are fulfilled, the
reduced density matrices of anyons can be expressed then in terms of
the reduced density matrices of fermions as
\begin{equation*}
(x_1,\cdots,x_n|\rho_n^a|x_1',\cdots,x_n')_{\pm}=A_{-\kappa}(x_1',
\cdots,x_n')B(x_1',\cdots,x_n')A_\kappa(x_1,\cdots,x_n)B(x_1,\cdots,x_n)
\end{equation*}
\be\label{rdmaf} \times\sum_{j=0}^\infty(-1)^j\frac{(1+e^{\pm
i\pi\kappa})^j}{j!}\int_{I_\pm}dz_{1}\cdots\int_{I_\pm}dz_{j}\
(x_1,\cdots,x_n,z_1,\cdots,z_j|\rho_{n+j}^f|x_1',\cdots,x_n',z_1,
\cdots,z_j) \ee
or bosons as
\begin{equation*}
(x_1,\cdots,x_n|\rho_n^a|x_1',\cdots,x_n')_{\pm}=A_{-\kappa}
(x_1',\cdots,x_n')A_\kappa(x_1,\cdots,x_n)
\end{equation*}
\be\label{rdmab} \times\sum_{j=0}^\infty(-1)^j\frac{(1-e^{\pm
i\pi\kappa})^j}{j!} \int_{I_\pm}dz_{1}\cdots\int_{I_\pm}dz_{j}\
(x_1,\cdots,x_n,z_1,\cdots,z_j|\rho_{n+j}^b|x_1',\cdots,x_n',z_1,
\cdots,z_j)\, . \ee
The subscript $\pm$ in these expressions specifies whether
$x_1,\cdots ,x_n,x_1',\cdots,x_n'$ are ordered as in $O_+$ or $O_-$.
\end{thm}

\begin{proof} The proof follows that of Lenard \cite{L1}, generalizing
it to the anyonic statistics. As a first step, we need a preliminary
result.

\begin{lem}\label{intQ}
For any symmetric function $f(z_1,\cdots,z_n)$ and a constant
$\alpha$,
\be\label{integral} \int_V dz_1\cdots\int_V dz_n\
\alpha^{\sigma(I_\pm)}f(z_1,\cdots,z_n)=\sum_{j=0}^n
C^n_j(-1+\alpha)^j\int_{I_\pm}dz_1\cdots\int_{I_\pm}dz_j \int_V
dz_{j+1}\cdots\int_V dz_n f(z_1,\cdots,z_n) , \ee
where $C^n_j=\frac{n!}{(n-j)!j!}$ and $\sigma(I_\pm)$ is the number of variables $z_1,\cdots, z_n$ contained in $I_\pm$.

\end{lem}

\begin{proof}

The L.H.S. of (\ref{integral}) can be written explicitly as
\be Q=\sum_{m=0}^n C_m^n \alpha^m\int_{I_\pm}
dz_1\cdots\int_{I_\pm}dz_m\int_{V\backslash I_{\pm}}dz_{m+1}
\cdots\int_{V\backslash I_{\pm}}dz_n f(z_1,\cdots,z_n)\, , \ee
and combined with an obvious relation $\int_{V\backslash
I_\pm}dz_i=\int_Vdz_i-\int_{I_\pm}dz_i\, ,\ (i=m+1,\cdots,n)$ can be
further transformed into
\be Q=\sum_{m=0}^n C_m^n
\alpha^m\sum_{k=0}^{n-m}C_k^{n-m}(-1)^k\int_{I_\pm}
dz_1\cdots\int_{I_\pm}dz_{m+k}\int_{V}dz_{m+k+1} \cdots\int_{V}dz_n
f(z_1,\cdots,z_n)\, . \ee
Collecting the terms in this expression with the same $j=m+k$, we
obtain the desired result
\begin{equation*}
Q=\sum_{j=0}^n C^n_j(-1+\alpha)^j\int_{I_\pm}dz_1\cdots \int_{I_\pm}
dz_j \int_V dz_{j+1}\cdots\int_V dz_n f(z_1, \cdots,z_n)\, .
\end{equation*}

\end{proof}

Now we can prove the Theorem \ref{main} starting with (\ref{rdmaf}).
Using the Anyon-Fermi mapping (\ref{AF}) we have

\begin{equation*}
(x_1,\cdots,x_n|\rho^a_n|x_1',\cdots,x_n')_{\pm}=\sum_{N=n}^\infty\sum_{\{\lam\}}
p^N_{\{\lam\}}\frac{N!}{(N-n)!} \int_V dz_{1}\cdots\int_V
dz_{N-n}C(x_1,\cdots,x_n,x_1',\dots,x_n')_{\pm}
\end{equation*}
\be \times\chi_{N}^{*f}(z_1,\cdots,z_{N-n},x_1',\cdots,x_n'|{\{\lam\}})
\chi_{N}^{f}(z_1,\cdots,z_{N-n},x_1,\cdots,x_n|{\{\lam\}})\, , \ee
where
\begin{eqnarray}
C(x_1,\cdots,x_n,x_1',\cdots,x_n')_{\pm}&=&A_{-\kappa}(x_1', \cdots
,x_n')B(x_1',\cdots,x_n')A_{\kappa}(x_1,\cdots,x_n)B(x_1,\cdots,
x_n) \nonumber\\
& &\times\prod_{j=1}^n\prod_{i=1}^{N-n} e^{-i\pi\kappa
\epsilon(z_i-x_j')/2 }e^{i\pi\kappa \epsilon(z_i-x_j)/2}
\epsilon(x_j'-z_i) \epsilon(x_j-z_i)\, .
\end{eqnarray}
One can see directly that
\be \prod_{j=1}^ne^{-i\pi\kappa \epsilon(z-x_j')/2}e^{i\pi\kappa
\epsilon(z-x_j)/2 } \epsilon(x_j'-z)\epsilon(x_j-z) =\left\{
\begin{array}{cl}-e^{\pm i\pi\kappa}, & z\;\; \mbox{ in } \; I_\pm\, ,\\
1,& z \, \mbox{ not in } \, I_\pm\, . \end{array}\right. \ee
This means that
\be
C(x_1,\cdots,x_n,x_1',\cdots,x_n')_{\pm}=A_{-\kappa}(x_1',\cdots,
x_n')B(x_1',\cdots,x_n') A_{\kappa}(x_1,\cdots,x_n)B(x_1,\cdots,x_n)
(-e^{\pm i\pi\kappa})^{\sigma'(I_\pm)}\, , \ee
where $\sigma'(I_\pm)$ is the number of variables $z_1,\cdots,
z_{N-n}$ in $I_\pm$. Applying now Lemma \ref{intQ} with
$\alpha=-e^{\pm i\pi\kappa}$, we obtain for the anyonic reduced
density matrices
\begin{equation*}
(x_1,\cdots,x_n|\rho^a_n|x_1',\cdots,x_n')_{\pm}=A_{-\kappa}
(x_1',\cdots,x_n')B(x_1',\cdots,x_n')A_{\kappa}(x_1,\cdots,x_n)
B(x_1,\cdots,x_n)\sum_{N=n}^\infty\sum_{\{\lam\}} p^N_{\{\lam\}}\frac{N!}{
(N-n)!} \end{equation*}
\be
\times\sum_{j=0}^{N-n}C^{N-n}_j(-1)^j(1+e^{\pm i\pi\kappa})^j
\int_{I_\pm}dz_1\cdots\int_{I_\pm}dz_{j}\int_Vdz_{j+1}\cdots\int_Vdz_{N-n}\
\chi^{*f}_{N,\{\lam\}}\chi^{f}_{N,\{\lam\}}\, . \ee
Interchanging the order of summations, one can notice that the sum
over $N$ and $\{\lam\}$ is precisely $\rho_{n+j}^f$. Therefore finally
\begin{equation*}
(x_1,\cdots,x_n|\rho_n^a|x_1',\cdots,x_n)_{\pm}=A_{-\kappa}(x_1',
\cdots,x_n')B(x_1',\cdots,x_n')A_\kappa(x_1,\cdots,x_n)
B(x_1,\cdots,x_n)
\end{equation*}
\be \times\sum_{j=0}^\infty(-1)^j\frac{(1+e^{\pm i\pi\kappa})^j}{j!}
\int_{I_\pm}dz_{1}\cdots\int_{I_\pm}dz_{j}\ (x_1,\cdots, x_n, z_1,
\cdots,z_j|\rho_{n+j}^f|x_1',\cdots,x_n',z_1,\cdots,z_j)\, . \ee

The proof of (\ref{rdmab}) is similar. In this case, we use the
Anyon-Bose mapping (\ref{AB}), and the $C_{\pm}$ function is
\begin{eqnarray}
C(x_1,\cdots,x_n,x_1',\cdots,x_n')_{\pm}&=&A_{-\kappa}(x_1',
\cdots,x_n')A_{\kappa}(x_1,\cdots,x_n) \nonumber\\
& &\prod_{j=1}^n\prod_{i=1}^{N-n} e^{-i\pi\kappa \epsilon(z_i
-x_j')/2} e^{i\pi\kappa\epsilon(z_i-x_j)/2}\, .
\end{eqnarray}
This means that we can use Lemma \ref{intQ} with $\alpha=e^{\pm
i\pi\kappa}$. Interchanging the order of summation and identifying
the bosonic reduced density matrices $\rho_{n+j}^b$ we obtain
(\ref{rdmab}).
\end{proof}

The results of Theorem \ref{main} do not depend on the statistical
ensemble used in the computation of the reduced density matrices as
long as the particles are subject to the hard-wall boundary
conditions making the state energies independent of the statistics.
They also do not depend on the form of the interparticle potential
beyond the need for the hard-core part which ensures that the
wavefunctions satisfy the hard-core condition. If both of these
conditions are satisfied, we can see from (\ref{rdmaf}) and
(\ref{rdmab}) that there is also no explicit dependence on the
length $L$ of the confining box $V$, and the results remain valid in
the thermodynamic limit $L\rightarrow\infty$.

\section{Impenetrable Free Case}\label{sect4}

The Anyon-Fermi relation derived above for the reduced density
matrices is particularly useful in the situation when the radius of
the hard-core interaction is vanishingly small, and no other
interactions are present. In this case, the fermionic problem is
identical to free fermions, since the hard-core potential of zero
radius effectively vanished due to antisymmetry of the
wavefunctions. The reduced density matrices $\rho_n^f$ coincide then
with those of free fermions \cite{L1} (see Appendix \ref{rdmff}):
\be (x_1,\cdots,x_n|\rho_n^f|x_1',\cdots,x_n')=\frac{1}{\pi^n}
\theta_T\left(\begin{array}{cc} x_1,\cdots,x_n\\
x_1',\cdots,x_n' \end{array}\right) , \label{4a} \ee
where $\theta_T(x,y)/\pi$ is the Fourier transform of the Fermi
distribution function:
\be \label{kernel} \theta_T(x,y)=\frac{1}{2}\int_{-\infty}^\infty
dk\frac{e^{ik(x-y)}}{1+e^{(k^2-h)/T}}\, . \ee At $T=0$ we have
\be\label{sinekernel} \theta_{0}(x,y)=\frac{\sin q(x-y)}{x-y}\, ,
\ee where $q=\sqrt{h}$ is the Fermi momentum.

Applying Theorem \ref{main} and (\ref{int}) we have
\be (x_1,\cdots,x_n|\rho_n^a|x_1',\cdots,x_n')_{\pm}= C
\sum_{j=0}^\infty(-1)^j\frac{(1+e^{\pm i\pi\kappa})^j}{j!}
\int_{I_\pm}dz_{1}\cdots\int_{I_\pm}dz_{j}\
\frac{1}{\pi^{n+j}}\theta_T\left(\begin{array}{cc}
x_1,\cdots,x_n,&z_{1},\cdots,z_{j}\\
x_1',\cdots,x_n',&z_{1},\cdots,z_{j}
\end{array}\right) ,
\ee
where
\begin{equation}
C(x_1',\cdots,x_n)\equiv A_{-\kappa}(x_1',
\cdots,x_n')B(x_1',\cdots,x_n')A_\kappa(x_1,\cdots,x_n)B(x_1,\cdots,x_n)
\label{cc} \end{equation}
and the subscript $\pm$ specifies particular ordering of
$x_1,\cdots,x_n,x_1',\cdots,x_n'$ as in Theorem \ref{main}. This
result can be rewritten in terms of Fredholm minors using
(\ref{Fredmin}) \be
(x_1,\cdots,x_n|\rho_n^a|x_1',\cdots,x_n')_{\pm}= \frac{1}{\pi^n}
C(x_1',\cdots,x_n) \det\left.\left(1-\gamma\hat
\theta_T^\pm \left|\begin{array}{c} x_1,\cdots,x_n\\
x_1',\cdots,x_n' \end{array}\right.\right) \right|_{\gamma=(1+e^{\pm
i \pi \kappa}) /\pi}\, , \ee
where the integral operator $\hat \theta^\pm_T$ with kernel
$\theta_T(x,y)$ is defined by its action on an arbitrary function
$f$:
\be (\hat \theta^\pm_T f)(x)=\int_{I_\pm}\theta_T(x,y)f(y)dy\,  \ee
Finally, introducing the resolvent kernel $\varrho_T^\pm(x,y)$ associated with the
the kernel $\theta_T(x,y)$, which satisfies
\be \varrho_T^\pm(x,y)-\frac{(1+e^{\pm
i\pi\kappa})}{\pi}\int_{I_{\pm}}\theta_T(x-z)\varrho_T^\pm(z,y)
dz=\theta_T(x-y)\, , \ee
and making use of (\ref{RK}), (\ref{multi}) and (\ref{2point}), we
obtain
\be
\la\fad(x_n')\cdots\fad(x_1')\fa(x_1)\cdots\fa(x_n)\ra_{T,h,\pm}=
\frac{ C(x_1',\cdots,x_n)}{\pi^n}
\varrho_T^\pm\left(\begin{array}{c} x_1,\cdots,x_n\\
x_1',\cdots,x_n' \end{array}\right) \left.\det\left(1-\gamma\hat
\theta_T^\pm\right)\right|_{\gamma=(1+e^{\pm i\pi  \kappa})/\pi} .
\ee
In the particular case of the simplest two-point correlator, this
expression reduces to
\be\label{2we}
\la\fad(x')\fa(x)\ra_{T,h,\pm}=\frac{1}{\pi}\varrho_T^\pm(x',x)
\left.\det\left(1-\gamma\hat \theta_T^\pm\right)\right|_{\gamma=
(1+e^{\pm i\pi\kappa})/\pi} , \ee
and gives the correlator of two anyonic fields in terms of the
Fredholm determinant of the integral operator $\hat\theta_T$ and its
resolvent kernel.

\section{Conclusions}

In this work, we have studied the field correlation functions
(reduced density matrices) for free impenetrable anyons in one
dimension and obtained the representation for these functions in
terms of the Fredholm minors of the integral operator with the
kernel given by (\ref{kernel}). This representation of the
correlation functions generalizes similar results \cite{L1} for
one-dimensional impenetrable bosons, and can be used to study the
asymptotic behavior of the correlation functions beyond the
approximation of conformal invariance. In the case of bosons, there
is also an alternative representation of the zero-temperature
one-particle reduced density matrix as the determinant of a Toeplitz
matrix, development of which was motivated by the study of momentum
distribution of bosons in the ground state \cite{G,L,S}. This
representation was extended recently to anyons by Santachiara {\it
et al.} \cite{SSC} in their study of the entanglement entropy of
impenetrable free anyons. As in the case of bosons, explicit
demonstration of the equivalence of the two representations for
anyonic correlation functions is an open problem.

\acknowledgments

This work was supported in part by the NSF grant DMS-0503712 and
DMR-0325551.

\appendix

\section{Fredholm Determinants}\label{Fredholm}

In this Appendix, we give a brief summary of results of Fredholm
theory of integral equations. For more details, see, e.g.,
\cite{text}. Consider the Fredholm equation of the second kind
\be f(x)-\gamma\int_{a}^b K(x,y)f(y)dy=g(x)\, , \ee
where the kernel $K(x,y)$ is a symmetric, bounded and continuous
function.

Defining operations with kernel $K(x,y)$ similarly to the usual
matrix operations:
\be K^n(x,y)=\int_a^bK(x,z)K^{n-1}(z,y)\ dz\, , \;\; \mbox{ with }
K^1(x,y)=K(x,y) \, , \ee
and
\be \mbox{Tr} K=\int_a^b K(x,x)\ dx\, , \;\; \mbox{Tr}
K^2=\int_a^b\int_a^b K(x,y)K(y,x)\ dx dy\, , \;\; \mbox{ and so on,}
\ee
we have the formulae that are useful for calculation of the Fredholm
determinant of the integral operator $1-\gamma\hat K$:
\be (1-\gamma\hat K)^{-1}=1+\gamma K^1+\gamma^2 K^2+\cdots \, , \ee
and
\be \ln \det(1-\gamma\hat
K)=-\sum_{n=1}^\infty\frac{\gamma^n}{n}\mbox{Tr} K^n \, . \label{a3}
\ee
Indeed, writing (\ref{a3}) as
\be \det(1-\gamma\hat K)= \prod_{n=1}^\infty \exp \{
-\frac{\gamma^n}{n}\mbox{Tr }K^n \} \, , \label{a4} \ee
and collecting terms of the same order in $\gamma$ one can see that
the determinant can be written conveniently as
\be \det\left(1-\gamma\hat K\right)=\sum_{n=0}^\infty(-1)^n
\frac{\gamma^n}{n!}\int_a^b dx_1\cdots\int_a^b dx_n K_n \left(
\begin{array}{c} x_1,\cdots,x_n \\x_1,\cdots,x_n \end{array}
\right) ,  \label{a5} \ee
where
\be \label{not} K_n \left(\begin{array}{c}
            x_1,\cdots,x_n\\
            y_1,\cdots,y_n
\end{array}\right) \equiv \det_{1\le j,k\le n} \left[ K(x_j,y_k)
\right]  .  \ee

The resolvent kernel $R(x,y)$ associated with the kernel $K(x,y)$ is
defined as $\hat R = (1-\gamma \hat K)^{-1} \hat K$, i.e.,
\be R(x,y)-\gamma\int_a^b\ K(x,z)R(z,y)dz=K(x,y)\, . \ee
If one introduces the determinants $R_n$ of kernels $R$ similarly to
(\ref{not}), an important relation can be proven to exist between
$R_n$ and the $r$-th Fredholm minor defined as a natural
generalization of Eq.~(\ref{a5}):
\be \label{Fredmin} \det\left(1-\gamma\hat K\left|\begin{array}{ccc}
y_1,&\cdots,&y_r \\
y_1',&\cdots,&y_r'
\end{array} \right.\right) =
\sum_{n=0}^\infty(-1)^{n}\frac{\gamma^{n}}{n!}\int_a^b dx_1\cdots
\int_a^b dx_n\ K_{n+r}\left(\begin{array}{cccccc}
y_1,&\cdots,&y_r,&x_1,&\cdots,&x_n\\
y_1',&\cdots,&y_r',&x_1,&\cdots,&x_n
\end{array}\right)\,.
\ee
The relation is \cite{H}
\be \label{RK} \det\left(1-\gamma\hat K\left|\begin{array}{ccc}
y_1,&\cdots,&y_r\\ y_1',&\cdots,&y_r' \end{array}\right. \right)
=\det\left(1-\gamma\hat K\right) R_n \left(\begin{array}{c}
y_1,\cdots,y_r \\ y_1',\cdots,y_r'
\end{array}\right) . \ee

\section{Anyonic Correlators}\label{ac}

In this Appendix, we prove Eq.~(\ref{eq100}) following the approach
used in \cite{AN} for calculation of the anyonic matrix elements. We
start with the simple case of the correlator
\be \la\Psi_2|\fad(x')\fa(x)|\Psi_2\ra\, , \ee
where we have omitted the quantum numbers $\{\lam\}$  unimportant
for the present computation. From (\ref{eq4e}) and (\ref{eq5e}), we
have
\be\label{B2} \la\Psi_2|\fad(x')\fa(x)|\Psi_2\ra=\frac{1}{2}\int
dz^2dy^2\ \chi_2^{*a}(y_1,y_2)\chi_2^a(z_1,z_2) \la
0|\fa(y_1)\fa(y_2)\fad(x')\fa(x)\fad(z_2)\fad(z_1)|0\ra\, . \ee
Defining for the moment
\be A=\la 0|\fa(y_1)\fa(y_2)\fad(x')\fa(x)\fad(z_2)\fad(z_1)|0\ra \,
, \ee
and using the commutation relation (\ref{com1}), $\fa(x)|0\ra=0$,
and $\la0|0\ra=1$ we obtain
\begin{eqnarray}\label{B4}
A&=&\la 0|\fa(y_1)\fa(y_2)\fad(x')\left[\fad(z_2)\fa(x)e^{-i\pi\kappa
\epsilon(x-z_2)}+\delta(x-z_2)\right]\fad(z_1)|0\ra\, ,\nonumber\\
&=&\la 0|\fa(y_1)\fa(y_2)\fad(x')\fad(z_2)\fa(x)\fad(z_1)|0\ra
e^{-i\pi\kappa\epsilon(x-z_2)} +\la0|\fa(y_1)\fa(y_2)\fad(x')\fad(z_1)
|0 \ra\delta(x-z_2)\, ,\nonumber\\
&=&\underbrace{\la
0|\fa(y_1)\fa(y_2)\fad(x')\fad(z_2)|0\ra}_{{\bf(a)}}
\delta(x-z_1)e^{-i\pi\kappa\epsilon(x-z_2)}
+\underbrace{\la0|\fa(y_1)\fa(y_2)\fad(x')\fad(z_1)|0 \ra}_{{ \bf
(b)}}\delta(x-z_2)\, .
\end{eqnarray}
Performing similar transformations we find that
\begin{eqnarray}\label{B5}
{\bf a}&=&\delta(y_1-x')\delta(y_2-z_2)e^{-i\pi\kappa\epsilon(y_2-x')}
+\delta(y_2-x')\delta(y_1-z_2)\, ,\nonumber\\
{\bf b}&=&\delta(y_1-x')\delta(y_2-z_1)e^{-i\pi\kappa\epsilon(y_2-x')}
+\delta(y_1-z_1)\delta(y_2-x')\, .
\end{eqnarray}
Substitution of (\ref{B4}) and (\ref{B5}) into (\ref{B2}) gives
\begin{eqnarray}
\la\Psi_2|\fad(x')\fa(x)|\Psi_2\ra &=&\frac{1}{2}\int
dz_1\left\{\chi_2^{*a}(x',z_1)\chi_2^a(x,z_1)
e^{-i\pi\kappa[\epsilon(z_1-x')+\epsilon(x-z_1)]}+\chi_2^{*a}(z_1,x')
\chi_2^a(x,z_1)e^{-i\pi\kappa\epsilon(x-z_1)}\right.\nonumber\\
& &\ \ \ \ \ \left.\chi_2^{*a}(x',z_1)\chi_2^a(z_1,x)e^{-i\pi\kappa
\epsilon(z_1-x')}+\chi_2^{*a}(z_1,x')\chi_2^a(z_1,x)\right\}\, .
\label{B6} \end{eqnarray}
Anyonic property of the wavefunctions (\ref{asdf}) together with its
complex conjugate
\be \chi_N^{*a}(z_1,\cdots,z_i,z_{i+1},\cdots,z_N)=e^{-i\pi\kappa
\epsilon(z_i-z_{i+1})} \chi_N^{*a}(z_1,\cdots,z_{i+1},z_{i},
\cdots,z_N)\, , \ee
means that (\ref{B6}) reduces to a simple form
\be \la\Psi_2|\fad(x')\fa (x)|\Psi_2\ra=2\int dz_1\
\chi_2^{*a}(z_1,x')\chi_2^a(z_1,x)\, . \ee
The generalization to the N-particle eigenstate is straightforward
and gives
\be \la\Psi_N|\fad(x')\fa(x)|\Psi_N\ra=N\int dz^{N-1}\
\chi_N^{*a}(z_1,\cdots,z_{N-1},x')\chi_N^a(z_1,\cdots,z_{N-1},x)\, .
\ee

\section{Reduced Density Matrices of Free Fermions}\label{rdmff}

The reduced density matrices for free 1D fermions were calculated in
the original paper of Lenard \cite{L1}. To make our discussion
self-contained, we provide here a sketch of the proof and the main
results in the notations that in general allow for an external
potential $U(z)$ acting on the particles.

We assume that the fermions are confined to the domain
$V=[-L/2,L/2]$, and have a complete set $\{ u_\lambda(z) \}$ of
normalized single-particle wavefunctions with energies
$\epsilon_{\lambda}$ in the potential $U(z)$. For instance, for
$U(z) \equiv 0$, and the "hard wall" boundary conditions at the
boundaries of the domain $V$,
\be u_\lambda(z)=\left\{\begin{array}{c}
\sqrt{\frac{2}{L}}\sin(\lambda z),\ \ \lambda=\frac{2\pi}{L},
\frac{4\pi}{L},\cdots\, ,\\ \sqrt{\frac{2}{L}}\cos(\lambda z),\ \
\lambda=\frac{\pi}{L},\frac{3\pi}{L},\cdots\, ,
\end{array}\right.  \ee
and, with appropriate conventions, $\epsilon_{\lambda}= \lambda^2$.
The $N$-body wavefunction of a stationary state is given by the
Slater determinant
\be \chi_N^f(z_1,\cdots,z_N|\{\lambda\})= \frac{1}{\sqrt{N!}}
\sum_{\pi\in S_N}(-1)^\pi\prod_{i=1}^N u_{\lambda_i}(z_{\pi(i)}) \,
, \ee
where the set $\{\lambda\}$ consists of non-coincident
single-particle states $\lambda_i$, and the energy eigenvalue is
$E(\{\lambda\})=\sum_{i=1}^N \epsilon_{\lambda_i}$. In the grand
canonical ensemble, the Gibbs measure is
\be p_{\{\lambda\}}^N=e^{hN/T}\frac{e^{-E(\{\lambda\})/T}} {Z(h,L,T)}\,
, \;\; \mbox{ with } \;\; Z(h,L,T)= \sum_{N= 0}^\infty
\sum_{\{\lambda\}}e^{hN/T}e^{-E(\{\lambda\})/T}, \label{c3} \ee
where $h$ is the chemical potential. Using the fact that, with an
extra factor $1/N!$ included to compensate for overcounting,
summation over $\{\lambda\}$ can be replaced with summation over
independent individual $\lambda_i$'s, one obtains the following
fundamental formula
\begin{eqnarray}\label{fun}
\sum_{\{\lambda\}}e^{hN/T} e^{-E(\{\lambda\})/T}
\chi_N^{*f}(z_1,\cdots,z_N|\{\lambda\})\chi_N^{f}(z_1',\cdots,z_N'|
\{\lambda\}) &=&\frac{1}{N!}\sum_{\pi\in S_N}(-1)^\pi\prod_{i=1}^N
F(z_i,z_{\pi(i)}') \\
&=&\frac{1}{N!}F_N\left(\begin{array}{c} z_1,\cdots,z_N \\
z_1',\cdots,z_N' \end{array}\right) .
\end{eqnarray}
Here we have used (\ref{not}) and
\be F(x,y)\equiv e^{h/T}\sum_\lambda
e^{-\epsilon_\lambda/T}u_\lambda^*(x)u_{\lambda}(y)\, .  \label{c4}
\ee
Equation (\ref{fun}) and proper normalization of the wavefunctions
$u_{\lambda}$ show that the fermionic statistical sum (\ref{c3}) can
be expressed as the determinant (\ref{a5}) of the integral operator
with kernel (\ref{c4}):
\be Z(h,L,T)=\sum_{N=0}^\infty\frac{1}{N!}\int_V dz_1 \cdots
\int_V dz_N\  F_N\left(\begin{array}{c} z_1,\cdots,z_N \\
z_1,\cdots,z_N
\end{array}\right) = \det(1+\hat F). \ee
Similarly, using the definition (\ref{Fredmin}) of Fredholm minor of
the same operator we see that
\begin{equation*}
\sum_{N=n}^\infty e^{hN/T}\sum_{\{\lambda\}} e^{-E(\{\lambda\})/T}
\frac{N!}{(N-n)!} \int_V dz_{1}\cdots \int_V dz_{N-n}\
\end{equation*}
\be \label{int} \times
\chi_N^{*f}(z_1,\cdots,z_{N-n},x_1,\cdots,x_n|\{\lambda\})
\chi_N^{f}(z_1,\cdots,z_{N-n},x_1',\cdots,x_n'|\{\lambda\}) = \det
\left( 1+\hat F\left|\begin{array}{c} x_1,\cdots,x_n\\
x_1',\cdots,x_n'
\end{array}\right.\right) , \ee
so that the reduced density matrix of the fermions can be expressed
as:
\be (x_1,\cdots,x_n|\rho_n^f|x_1',\cdots,x_n')=\det\left(1+\hat
F\left|\begin{array}{c} x_1,\cdots,x_n\\ x_1',\cdots,x_n'
\end{array}\right.\right)/\det(1+\hat F) .
\ee
The relation (\ref{RK}) for Fredholm minors means that this result
can be expressed simply in terms of the resolvent kernel
$\theta_T(x,y)/\pi $ associated with kernel $F(x,y)$ (\ref{c4})
(factors of $\pi$ are chosen so that the notations are the same as
in the main text):
\be (x_1,\cdots,x_n|\rho_n^f|x_1',\cdots,x_n')= \frac{1}{\pi^n}
\theta_T\left(\begin{array}{c} x_1,\cdots,x_n\\
x_1',\cdots,x_n'
\end{array}\right) . \ee
In the thermodynamic limit with no external potential, $U(z) \equiv
0$, the resolvent kernel $\theta_T$ is given by
\be \lim_{L\rightarrow\infty} \theta_T(x,y)= \frac{1}{2}
\int_{-\infty}^\infty dk\frac{e^{ik(x-y)}}{1+e^{(k^2-h)/T}}\, . \ee
%



\end{document}